\documentclass[envcountsect,envcountsame]{llncs}
\usepackage{amssymb,stmaryrd}
\usepackage{amsmath}
\usepackage{epsfig}
\usepackage{enumerate}
\usepackage{algorithmic}
\usepackage[ruled,section]{algorithm}

\algsetup{indent=2em}
\renewcommand{\algorithmicendif}{\textbf{fi}}

\newcommand{\algorithmicdone}{\textbf{done}}
\newcommand{\OLFORALL}[2]
{\STATE \algorithmicfor \mbox{ $\forall$#1 }\algorithmicdo\mbox{ #2 }\
  \algorithmicdone}

\newcommand{\OLIF}[2]
{\STATE \algorithmicif \mbox{ #1 }\algorithmicthen\mbox{ #2 }\algorithmicendif}
\newcommand{\OLELIF}[3]
{\STATE \algorithmicif \mbox{ #1 }\algorithmicthen\mbox{ #2 }\
  \algorithmicelse \mbox{ #3} \algorithmicendif}

\newcommand{\is}{\leftarrow}

\def\ca#1{{\cal#1}}
\let\sem\setminus

\title{\mbox{Detours in Scope-Based Route Planning}%
\thanks{Research supported by the Czech Science Foundation, grant P202/11/0196.}
}

\author{Petr Hlin\v{e}n\'y ~\and~ Ondrej Mori\v{s}}

\institute{Faculty of Informatics, Masaryk University \\
  Botanick\'a 68a, 602 00 Brno, Czech Republic  \\
  \email{hlineny@fi.muni.cz, xmoris@fi.muni.cz}}

\begin{document}

\maketitle

\begin{abstract}
We study a dynamic scenario of the static route planning problem in road 
networks. Particularly, we put accent on the most practical dynamic case
-- increased edge weights (up to infinity). We show how to enhance the 
scope-based route planning approach presented at ESA'11, \cite{HM2011A} 
to intuitively by-pass closures by detours. Three variants of a detour 
``admissibility'' are presented -- from a~simple one with straightforward
implementation through its enhanced version to a full and very complex 
variant variant which always returns an optimal detour. 
\end{abstract}

\section{Introduction}
\label{sec:introduction}
 
Route planning has  many important everyday applications. In fact, it is 
a single pair shortest path (SPSP) problem in real-world road networks. There 
are several specific variations of this problem. Particularly, there are two 
route planning \emph{variants} and two \emph{scenarios}. First, in a 
time-independent variant which is studied in our paper, chosen cost function 
does not depend on time while in a more challenging time-dependent variant it 
does -- i.e., costs of a route depends on departure time. Secondly, a~route 
planning scenario is static if a road network is fixed while in a dynamic 
scenario costs or even an overall is changing predictably (traffic jams, 
turn-angle limits) or unexpectedly (car accidents). 

Complexity of the problem depends varies in different combinations of a~variant 
and a scenario. The very basic static time-independent route planning received
a lot of attention during the last decades and hence we focus on dynamic 
time-independent route planning which is more complicated but also more 
realistic. Furthermore, we believe it has one of the best practical motivations
among all route planning variations. Unfortunately, neither classical graph 
algorithms for SPSP problem such as Dijkstra's \cite{Dijkstra1959}, A* 
\cite{Hart1972} algorithm nor their dynamic adaptions \cite{Cooke1966} are well 
suitable even for time-independent route planning. It is mainly because graphs 
representing real-world road network are very huge. Clearly, a~feasible solution 
lies in splitting algorithms in two phases since route planning problem 
instances are mostly solved on a single road network. First, in a 
\emph{preprocessing} phase, we invest some time to exploit auxiliary data from 
a~road network. Secondly, these auxiliary data are utilized to improve both 
time and space complexity of subsequent instances of the problem -- 
\emph{queries}.

\subsubsection{Related Work.} 
This technique led to several very interesting static approaches in the last 
decade, we refer our reader to surveys \cite{Schultes2008,Delling2009}. 
Unfortunately, most of the recently developed techniques require a rather 
unrealistic assumption -- static road networks. Modifying static approaches 
for dynamic road networks is much harder than one might expect because dynamic
changes invalidate preprocessed data. Nevertheless, there are modifications 
already proven to work in a dynamic scenario such as, for instance, 
highway-hierarchies \cite{Sanders2006,Schultes2007}, ALT 
\cite{Goldberg2005B,Delling2007} or geometric containers 
\cite{Wagner2004,Wagner2005}. Even though this paper does not deal with 
time-dependent route planning, we refer to \cite{Delling2009} for more 
details on this complicated topic.

Recently, in order to fill a gap between a variety of exact route planning
approaches, we have published \cite{HM2011A} a different novel approach
aimed at ``reasonable'' routes.  It is based on a concept of scope, whose
core idea can be informally outlined as follows: The edges of a road network
are associated with a {\em scope} map such that an edge $e$ assigned scope 
$s_e$ is admissible on a route $R$ if, before or after reaching $e$, such 
$R$ travels distance less than a value associated with $s_e$ on edges with 
scope higher than~$s_e$.  The desired effect is that low-level roads are 
fine near the start or target positions, while only roads of the highest 
scope are admissible in the long middle sections of distant routing queries. 
Overall, this nicely corresponds with human thinking of intuitive routes, 
and allows for a very space-efficient preprocessing, too.

\paragraph{New Contribution.}
We present a dynamic adjustement of aforementioned admissibility concept 
along with a simple modification of scope-based Dijkstra's algorithm.
We allow a reasonably small number of road closures in a road network and 
out approach can be straightforwardly generalized also for slowed-down 
roads. We present three definition of admissile detours -- from the simple 
one with efficient implementation with time and space complexity of the 
original scope-based Dijkstra's algorithm, though its enhanced version to
the full and most complex one which always ensures the existence of a 
detour.  We have briefly experimentally evalueted implementation of the 
first and second definitions with very promising results and have also 
incorporated these algorithms into scope-based route planning approach.

\section{Fundamentals}
\label{sec:fundamentals}

\paragraph{Graphs and Walks.} 
A {\em directed graph} $G$ is a pair of a finite set $V(G)$ of vertices 
and a finite multi-set $E(G) \subseteq V(G) \times V(G)$ of edges. A 
\emph{walk} $P \subseteq G$ is an alternating sequence $(u_0,e_1,u_1,\ldots,
e_k,u_k)$ of vertices and edges of $G$ such that $e_i = (u_{i-1},u_i)$ for 
$i = 1,\ldots,k$. To point out the start vertex $u=u_0$ and the end $v=u_k$, 
we say $P$ is a $u$-$v$ walk. A {\em subwalk} $Q\subseteq P$ is $Q=(x=u_i,
e_{i+1},\ldots,u_j=y)$ for some $0\leq i\leq j\leq k$, and it is referred to 
as $Q=P^{xy}$ (for simplicity, possible ambiguity with exact reference to the 
position of $x,y$ in $P$ is neglected). The \emph{weight} of a walk $P 
\subseteq G$ w.r.t. a weighting $w: E(G) \mapsto \mathbb{R}$ of $G$ is 
defined as $|P|_w=w(e_1)+w(e_2)+\dots+w(e_k)$ where $P=(u_0,e_1,\ldots,e_k,u_k)$. 
An {\em optimal walk} between two vertices achieves the minimum weight over 
all walks.

\paragraph{Road Networks.}
A {\em road network} is referred to as a pair $(G,w)$. Naturally, $G$ is 
a~directed graph $G$ such that the junctions are represented by $V(G)$ and 
the roads by $E(G)$). Moreover, $G$ is assigned a {\em non-negative} edge 
weighting $w$ representing chosen cost function (or a suitable combination 
of cost functions). 

\paragraph{Problem Formulation.}
Given a road network $(G,w)$ and start and target vertices $s,t \in V(G)$, 
find a walk from $s$ to $t$ \emph{optimal with respect to given optimality 
criteria}. Regarding a weighting function, there are two \emph{variants} of 
the problem -- in a \emph{time-independent} one, $w : E(G) \mapsto 
\mathbb{R}^+_0$, while in \emph{time-dependent}, $w : E(G) \times \mathbb{N}_0 
\mapsto \mathbb{R}^+_0$ depends on discrete departure time slots (or their
linear interpolation, see \cite{Delling2009}). Moreover, there are two route 
planning \emph{scenarios} -- in a \emph{static} one, both $G$ and $w$ are 
fixed during all queries. On the other hand, $G$ and $w$ may change either 
predictably (e.g. rush hours) or unexpectedly (e.g. car accidents) in a 
\emph{dynamic} scenario; the updated road network is denoted by~$(G^*,w^*)$.

\medskip
In this paper, we study a specific version of the time-independent {\em dynamic}
route planning problem: the underlying graph $G$ remains 
static and $w$ is only allowed to increase. Formally, $G^* = G$ and $w^*: E(G) 
\mapsto \mathbb{R}_0^+ \cup \{\infty\}$, $w^*(e) \ge w(e)$ for all $e \in E(G)$ 
(particularly, $w^*(e) = \infty$ implies that $e$ is ``closed''). 
The choice of this version is driven by typical real-world situations. 
To advocate this simply and 
informally, occasions on which a road is improved or a new one built are much 
less frequent than temporary road closures. We thus for simplicity omit the 
possibility of adding new edges to $G$ in this paper, though we keep in mind 
that locally adding an edge may be necessary, e.g., to designate a detour.

We note again that we do not deal with the full power of time dependent planning
as~\cite{Delling2009}, mainly due to a totally different conceptual view of 
the latter and also due to the fact that it is not well understood yet.
Still, our dynamic scenario can incorporate some aspects of time-dependent
weghting function, namely on-demand reflection of rush hours on
selected edges (i.e., busy roads).

Finally, we assume familiarity with classical Dijkstra's algorithm and its 
bidirectional variants (otherwise see Appendix \ref{app:dijkstra}) for shortest 
paths. 

\subsection{Scope-Based Route Planning in a Nutshell}
\label{sec:scope}

A simplified version of the recently introduced \emph{scope 
concept}~\cite{HM2011A} is very briefly recapitulated here. We strongly 
recommend reading the original paper \cite{HM2011A} for better understanding 
and more detailed treatment. Due to lack of space, many details are omitted 
here. The purpose of introducing scope has been twofold: to capture in a 
mathematically rigorous way a vague meaning of ``comfort and intuitivness'' 
of a route, and at the same time allow for more memory efficient preprocessing
of the static road network data for very fast subsequent queries. It works
best with a cost function correlated with travel time. 

\begin{definition}[Scope \cite{HM2011A}]
\label{def:scope}
Let $(G,w)$ be a road network. \emph{A scope mapping} is defined as ${\cal{S}}: 
E(G) \mapsto \mathbb{N}_0 \cup \{\infty\}$ such that $0,\infty \in 
Im({\cal{S}})$. Elements of the image $Im({\cal{S}})$ are called \emph{scope 
levels}. Each scope level $i\in Im({\cal{S}})$ is assigned a constant value of 
\emph{scope} $\nu^{\cal{S}}_i \in \mathbb{R}_0 \cup \{\infty\}$ such that 
$0 = \nu^{\cal{S}}_0 < \nu^{\cal{S}}_1 < \cdots < \nu^{\cal{S}}_\infty = \infty$.
\end{definition}

In practice there are only a few scope levels in $Im({\cal{S}})$ (say,~5). 
The desired effect, as formalized next, is in \emph{clever} using low-level 
roads only near the start or target positions until higher level roads 
become widely available. For that one has to count how much has been travelled 
along a given walk on edges of higher level (Def.~\ref{def:Sdraw}), 
and do not admit lower-level edges further on (Def.~\ref{def:stadmissible},
iii.).

\begin{definition}[Scope $\ca S$-draw]
\label{def:Sdraw}
Let $(G,w)$ be a road network and a scope mapping $\cal{S}$.
The $\ca S$-draw value of a walk $P\subseteq G$ is a vector 
$draw^{\ca S}(P)=\vec\sigma$ indexed by $Im({\cal{S}})$ 
such that $\vec\sigma_\ell=
	\sum_{f \in E(P),\, {\cal{S}}(f)>\ell}\, w(f)$
for $\ell\in Im({\cal{S}})$.
\end{definition}

For practical applications, the formula for $draw^{\ca S}(P)$ is expanded with 
a so called {\em turn-scope handicap} \cite{HM2011A} penalizing for missed 
higher-level edges.

\begin{definition}[Admissibility \cite{HM2011A}]
  \label{def:stadmissible}
  Let $(G,w)$ be a road network. Consider a walk  $P=(s=u_0,e_1,\dots e_k,u_k=t)
  \subseteq G$ from the start $s$ to the end~$t$.
  \begin{enumerate}[a)]\vspace{-4pt}
  \item  An edge $e = (v_1,v_2) \in E(G)$ is \emph{$x$-admissible} in $G$ for 
    a scope mapping $\cal{S}$ if, and only if, there exists a walk $Q 
    \subseteq G - e$ from $x\in V(G)$ to $v_1$ such that
    \begin{enumerate}[i.]\parskip1pt
    \item each edge of $Q$ is recursively $x$-admissible in $G - e$ for $\cal{S}$,
    \item $Q$ is optimal subject to (1), and
    \item for $\ell={\cal{S}}(e)$ and $\vec\sigma=draw^{\ca S}(Q)$,
      it is $\vec\sigma_\ell\leq \nu^{\cal{S}}_{\ell}$.
    \end{enumerate}
    \smallskip
  \item Whole $P$ is \emph{$s$-admissible} in $G$ if every $e_i \in E(P)$ is 
    $s$-admissible in $G$;
    \smallskip
  \item\label{it:revR}
    and $P$ is \emph{$\ca S$-admissible} 
    (with implicit respect to $s,t$) if there exists $0\leq j\leq k$ such that 
    every $e_m \in E(P)$, $m\leq j$, is $s$-admissible in $G$, and the reverse 
    of every $e_m\in E(P)$, $m>j$, is $t$-admissible in reverse $G^R$.
  \end{enumerate}
\end{definition}

Note the last part \ref{it:revR}) of Definition~\ref{def:stadmissible}---
there and further on we often use a simplificating term ``{\em in reverse}''
to refer to the network $(G^R,w)$ obtained by reversing all edges of $G$,
and to exchanged start and end $t,s$. This is to make our definitions 
symmetric from the viewpoint of $s$ as from~$t$.

\begin{remark}
A vertex $v \in V(G)$ is \emph{$s$-saturated} (for $s \in V(G)$) if 
$[draw^{\ca S}(P)]_{\ell} > \nu^{\cal{S}}_{\ell}$ for $\ell < \infty$ where $P$ 
is an optimal $s$-admissible $s-v$ walk. In other words, a vertex is $s$-saturated 
if any $s$-admissible edge leaving the last vertex of $P$ has the scope level $\infty$.
\end{remark}

\subsubsection{Static $\cal{S}$-Dijkstra's Algorithm of \cite{HM2011A} 
(see Appendix~\ref{app:sdijkstra}).}

Having a definition of admissible walks, one needs also a corresponding 
route planning algorithm. The seemingly complicated Def.~\ref{def:stadmissible} 
can actually be smoothly and simply integrated into traditional Dijkstra's 
or A* algorithms and their bidirectional variants. 

\begin{itemize}
\item For each scanned vertex $v$, a track of the best value $\vec\sigma[v]$ 
  of $\ca S$-draw is kept.
\item An edge $e$ leaving $v$ is relaxed only if $\vec\sigma_{{\cal S}(e)}[v]\leq 
  \nu^{\cal{S}}_{{\cal S}(e)}$ (cf.~Def.~\ref{def:stadmissible}, iii.).
\end{itemize}

\begin{theorem}[\cite{HM2011A}]
  \label{thm:SDijkstra}
  {\em$\cal{S}$-Dijkstra's algorithm} (uni-di\-rectional), for a road 
  network $(G,w)$, a scope mapping ${\cal{S}}$, and a start vertex 
  $s \in V(G)$, computes an optimal $s$-admissible walk from $s$ to 
  every $v \in V(G)$ in time ${\cal O}\big(|E(G)|\cdot|Im({\cal{S}})|+
  |V(G)|\cdot\log |V(G)|\big)$.
\end{theorem}

An optimal $\ca S$-admissible $s$-$t$ walk in $(G,w)$ is then found
by a natural bidirectional application of Theorem~\ref{thm:SDijkstra},
which also allows for a very efficient preprocessing of the road network,
as detailed in \cite{HM2011A}.

\section{Detours -- on the Price of Admissibility}
\label{sec:detours}

In the dynamic scenario a static $\cal{S}$-Dijkstra's algorithm may badly fail.
Imagine a driver approaching a restricted tunnel (e.g. by a car 
accident) such that it can be bypassed on low-level mountain roads only. 
What would a driver do? 

She could drive through this restrictions according to her original route plan 
and accept increased cost (if possible). However, there might be a better route.
Notice that a re-planning her route from scratch might not be possible due to 
temporarily invalidated preprocessed data. The best intuitive solution for her 
is to slip off the original route (even ahead of the restricted tunnel) and use 
a detour by re-allowing the use of low-level (i.e., inadmissible in the ordinary 
setting) mountain road nearby this restriction. She still wants to minimize 
costs of such detour and drive comfortably within the margins of such adjusted 
scope admissibility view. On the other hand, static Def.~\ref{def:stadmissible}
would not allow the aforementioned detour for natural reasons (unless the closure is
near the start or target position); the static scope mapping simply cannot 
account for such unexpected closures in advance. Yet there is a good solution 
which extends the very nice properties of static scope to the considered 
dynamic scenario.

\subsection{General Strategy for Avoiding Closures}

As mentioned in Section~\ref{sec:fundamentals}, a formal view of this dynamic 
scenario is that the original static road network $(G,w)$ with $\ca S$ is 
replaced by $(G,w^*)$ where $w^*$ {\em increases the weight} of some edges (up
to $\infty$), while $\ca S$ and $G$ stay the same. Let $C=\{e\in E(G): 
w^*(e)>w(e)\}$. For simplicity, we further assume $w^*(e)=\infty$ for $e\in C$
and call $C$ the set of {\em (road) closures}, but a generalization to 
arbitrary weight increase $w^*(e) > w(e)$ is straightforward. Hence, 
from now on, we focus only on $w^*$ and the closure set $C\subseteq E(G)$
which is {\em assumed relatively small}.

The mathematical task is to relax the meaning of scope admissibility close 
to the edges of~$C$. For that we slightly extend the definitions of 
Section~\ref{sec:scope}: We say $\vec\omega$ is an {\em$\ca S$-vector} if 
$\vec\omega$ is indexed by $Im({\cal{S}})$. As in 
Definition~\ref{def:stadmissible}\,b), a walk $P=(s=u_0,e_1,\dots u_k=t)$ 
is called \emph{$(s,\vec\omega)$-admissible} if the condition (iii) 
newly reads $\vec\sigma_\ell+\vec\omega_\ell \leq \nu^{\cal{S}}_{\ell}$. 
Similarly, as in Def.~\ref{def:stadmissible}\,b), this $P$ is {\em$\ca 
S$-admissible when amended} with the initial and/or final $\ca S$-vectors 
$\vec\omega^s$, $\vec\omega^t$, if there exists $0\leq j\leq k$ such that 
every $e_i \in E(P)$, $i\leq j$, is $(s,\vec\omega^s)$-admissible in $G$, 
and the reverse of every $e_i\in E(P)$, $i>j$, is 
$(t,\vec\omega^t)$-admissible in reverse $G^R$. This definition simply 
captures a possibility that some of the $\ca S$-draw value has already 
been used (exhausted) before entering~$P$.

As mentioned at the beginning of this section, we would like to allow
limited use of inadmissible (either for $s$ or $t$) edges near the closures. 
The crucial question is to decide which inadmissible edges should be
additionally allowed. The first step is to specify at which vertices a closure 
affects an $s$-$t$ route.

\begin{definition}[$C$-obstructed vertex]
  \label{def:c-obstructed}
  Let $(G,w)$ be a road network with a scope mapping $\ca S$,  
  $C\subseteq E(G)$ a set of closures and $s,t \in V(G)$.   
  \begin{enumerate}[i.]
  \item A vertex $d\in V(G)$ is {\em$C$-obstructed} for the initial 
    $\ca S$-vector $\vec\omega^s$ and target $t$ if there exists a 
    $d$-$t$ walk $Q\subseteq G$ which is optimal $\ca S$-admissible when 
    amended with initial $\vec\omega^s$, such that $Q$ contains some (any)
    edge $(z,v)\in C$. 
    \smallskip
  \item A \emph{$C$-obstruction state} of $d$ is the $\ca S$-vector
    $\vec\sigma$ such that, for each scope level $\ell \in Im({\cal{S}})$,
    $\vec\sigma_\ell$ is the minimum of $\big[draw^{\ca S}(Q^{dz})
    \big]_\ell$ over all walks $Q$ and $z$ as from (i.),
    where $Q^{dz}$ is the (shortest) $d$-$z$ subwalk of~$Q$.
  \item Moreover, a \emph{$C$-obstruction level} of $d$
    (again for initial $\vec\omega^s$ and target $t$) is the minimum scope 
    level $\ell \in Im({\cal{S}})$ such that $\vec\sigma_\ell\leq\nu^{\cal{S}}_\ell$.
  \end{enumerate}
  If $\vec\omega^s=(\infty,\dots,\infty)$ (the most restrictive case),
  then we shortly say that $d$ is {\em$C$-obstructed} for the target $t$.
  Analogously, $d$ is {\em$C$-obstructed} for the start $s$ and final 
  $\ca S$-vector $\vec\omega^t$ if (i.--\,iii.) holds in reverse. 
\end{definition}


Note that typically only one optimal $d$-$t$ walk $Q$ exists in (i.), 
but for formal correctness of the definition we have to range over all
possible ones in (ii.).

For an informal explanation, $d$ is $C$-obstructed when 
an edge $e \in C$ (a closure) affects an optimal $\cal{S}$-admissible 
walk from $d$ to the target $t$, and $e$ is not ``far away'' from $d$ 
wrt.\ the $\ca S$-draw value on level $\ell$; or this symmetrically 
happens in reverse~$G^R$. The role of $\vec\omega^s,\vec\omega^t$ in 
Def.~\ref{def:c-obstructed} is purely technical (to capture $\ca S$-draw 
value travelled prior to approaching $d$ in certain situations), and one 
may simply ignore it for getting the general informal picture.

\subsection{Simple Detours}
\label{sec:simple-detours}

In this section, we can finally define the simplest version of $C$-detour 
admissibility -- i.e., the relaxation of traditional $\ca S$-admissibility 
in the presence of closures. Informally, an $s$-$t$ walk is simple $C$-detour
$\ca S$-admissible if it avoids all the closed roads in $C$;
and, in addition to ordinary $\ca S$-admissibility
(Def.~\ref{def:stadmissible}),
certain ``detour permits'' are issued near those edges of $C$ which
(potentially) obstruct an optimal $s$-$t$ walk.
A simple detour permit just allows limited local use of edges of low
scope-level (which would not be permitted otherwise by
Def.~\ref{def:stadmissible}), until non-closed higher level edges become
available again. These permits (and subsequent detours) can be repeated
along an admissible walk, as needed by further closures.
The formal definition is as follows.

\begin{definition}[Simple $C$-detour $\ca S$-admissibility]
  \label{def:simple-c-detour}
  Let $(G,w)$ be a road network, $\ca S$ a scope mapping on it,
  and $C\subseteq E(G)$ a set of road closures. An $s$-$t$ walk 
  $P=(s=u_0,e_1,\dots e_k,u_k=t)\subseteq G$ is {\em simple $C$-detour 
    $\ca S$-admissible} if $E(P)\cap C=\emptyset$
  and there exists $0\leq j\leq k$ such that, for 
  each $e_m\in E(P)$, (at~least) one of the following holds:
  \begin{enumerate}[i.]
    \parskip 2pt
  \item $m\leq j$ and $e_m$ is $s$-admissible in $G$ for $\ca S$.
  \item $m>j$ and $e_m$ is $t$-admissible in reverse $G^R$ for $\ca S$.
  \item $\ca S(e_m)=\ell<\infty$ and there exists $0\leq i< m$ such that;
    \begin{itemize}
    \item
      the vertex $d=u_i$ of $P$ is $C$-obstructed for the target~$t$ --
      optionally also for initial $draw^{\ca S}(Q)$
      where $Q$ is an optimal $s$-admissible $s$-$d$ walk%
      \footnote{This part with $Q$ actually applies only when $d$ is not 
        $s$-saturated, i.e., nearby $s$.},
      and
    \item
      the level of obstruction of $d$ is $\ell < \infty$
      and no vertex among $u_{i+1},\dots,u_{m-1}$ is 
      left by an edge $f\not\in C$ in $G$ of $\ca S(f)>\ell$.
    \end{itemize}
  \item Or, stating briefly, (iii) holds in reverse for some $m\leq i<t$.
  \end{enumerate}
\end{definition}
In this definition, points (i.), (ii.) refer to ordinary $\ca S$-admissibility. 
Point (iii.) then permits use of lower-level edges since an obstructed vertex 
$d$ till higher-level becomes available. Actually, there is a minor issue of 
aforementioned Def.~\ref{def:simple-c-detour} left for discussion (and 
resolution) to Section~\ref{sec:enhanced-c-detour}.


\subsubsection{Simple $C$-Detour $\cal{S}$-Dijkstra's Algorithm.} 
Even though Def.~\ref{def:simple-c-detour} might seem complicated, it can be 
implemented straightforwardly by running just a single bidirectional 
$\ca S$-Dijkstra's algorithm. Due to lack of space we present the idea
of our algorithm and omit implementation details in this paper.
Let $(G,w)$ be an original road network with a scope mapping $\cal{S}$, $s,t 
\in V(G)$ start and target vertices, $w^*$ a changed weighing and $C$ a set of 
closures. 

\begin{enumerate}
\parskip 3pt
\item \emph{Static initialization}
  \smallskip

  Static $\cal{S}$-Dijkstra's algorithm is executed in $(G,w)$ bidirectionally
  to find optimal $\ca S$-admissible $s-t$ walk $P$. If $w(P) = w^*(P)$ then our
  algorithm terminates since no detour is needed. Otherwise an optimal simple 
  $C$-detour $\ca S$-admissible $s-t$ walk $Q$ is to be found. If $w^*(Q) < 
  w^*(P)$ then $Q$ is returned, otherwise $P$.
  
\item \emph{Identifying $C$-obstructed vertices}
  \smallskip

  In order to find $Q$, we must identify $C$-obstructed for both $s$ and $t$ and 
  (possibly) their corresponding $\vec\omega_t,\vec\omega_s$ (for obstructions
  close the the start or target). This is done again by running static 
  $\ca S$-Dijkstra's algorithm executed bidirectionally from $s$ and $t$, but 
  now in updated $(G,w^*)$ and $(G,w^*)^R$, respectively. Moreover, 
  the algorithm does not terminate until both search queues are empty. Consider
  the forward search, the reverse is analogous. 
  
  When $\ca S$-Dijkstra's algorithm scans a vertex $u$ such that an edge 
  $(\pi[u],u)$ from its predecessor $\pi[u]$ is a closure, $u$ is 
  $C$-obstructed Def.~\ref{def:c-obstructed}. All successors of $u$ are 
  $C$-obstructed as well. For such vertices $v$ a reference to the end vertex 
  of the nearest closure on their $s$-admissible $s-v$ walk is stored. Using 
  this reference we can easily determine an obstruction state in constant time.
  Using this process we identify all $C$-obstructed vertices for $s$ with
  final $\vec\omega_t=(\infty, ..., \infty)$. To get the other 
  $C$-obstructed vertices for $s$ we must combine forward and reverse searches 
  as follows. When there is a vertex $w$ scanned in both direction such that 
  there is a closure $c=(x,y)$ on $w-t$ $t$-admissible walk then all 
  non-$t$-saturated vertices on $y-t$ $t$-admissible walk are $C$-obstructed
  $s$ with final final $\vec\omega_t$ given by their $\ca S$-draw values in
  the reverse search.

\item \emph{Permitting not $\ca S$-admissibile edges}
  \smallskip

  Now we have identified $C$-obstructed vertices for both $s$ and $t$ together
  with their $C$-obstruction states and for those close enough to $s$ or $t$ we also
  know initial and final vectors $\vec\omega^t$ and $\vec\omega^s$. From obstruction
  states we can easily determine obstruction levels of these vertices, $\ca S$-draw 
  values, predecessors and distance estimates from $s$ (or to $t$, respectively). 
  Again, consider the forward direction. If a $C$-obstructed vertex $v$ for $t$ has 
  obstruction level $\ell$ smaller than $\infty$, gets \emph{a permit} for relaxing
  any of its outgoing edges $e=(v,w)$ of level ${\ca S}(e) = \ell$ without updating 
  $\ca S$-draw value of $w$ and $v$ is pushed to the search queue with updated 
  distance estimate and parent. The same holds for all edges of $w$ and next 
  successors. Once there is an edge of scope higher than $\ell$ outgoing $w$, 
  permitting stops. This is done for all $C$-obstructed vertices for both $s$ and 
  $t$. 

\item \emph{Completing simple $C$-detour $\ca S$-admissible $s-t$ walk}
  \smallskip

  Finally, we keep static $\ca S$-Dijkstra's algorithm running with all 
  auxiliary data structures computed so far in $(G,w^*)$, the resulting $s-t$ 
  walk is $Q$. 

\end{enumerate}

We would like to emphasize that all aforementioned steps can be done in 
a~single bidirectional execution of $\ca S$-Dijsktra's algorithm. Hence 
running time of aforementioned algorithm remains in 
${\cal O}\big(|E(G)|\cdot|Im({\cal{S}})|+|V(G)|\cdot\log |V(G)|\big)$.

\subsection{Enhanced Detours}
\label{sec:enhanced-c-detour}

Unfortunately, aforementioned Definition~\ref{def:simple-c-detour} has a minor 
practical drawback which can cause that in some very specific cases there is no 
simple $C$-detour $\ca S$-admissible $s$-$t$ walk even though a $C$-avoiding
$s$-$t$ walk exists:
\begin{itemize}
\item
Imagine a one-directional road segment $f'=(d,d')$ just preceding $f\in C$
(more generally, a local map area which can be left only through $f$).
Then one cannot take a detour from $d'$, and a lower-scope detour edge from
$d$ may not be allowed by Def.~\ref{def:c-obstructed}.iii.
\item
Hence in such situation this $f'$ is {\em effectively closed} as well,
and we can capture this in a supplementary definition which
consequently resolves the issue:
\end{itemize}

\begin{definition}[Quasi-Closure]
Let $(G,w)$ be a road network, $C \subseteq E(G)$ a set of closures and
$t \in V(G)$ the target vertex. An edge $(u,v) \in E(G)\sem C$ is 
\emph{$C$-quasi-closed for $t$} 
if there is no $\ca S$-admissible $u$-$t$ walk starting with $(u,v)$
in the subnetwork $(G-C,w)$.

A $C$-quasi-closed edge for the start $s$ is defined analogously in reverse.
\end{definition}

For a set $C$ of closures, we denote by $C^*$ its {\em qc-closure},
i.e., the least fixed point of the operation of adding $C$-quasi-closed
edges for $t$ or $s$ to~$C$. We then
easily amend the definition of simple detour admissibility as follows.

\begin{definition}[Enhanced $C$-detour $\ca S$-admissibility]
\label{def:enhanced-detour}
  Let $(G,w)$ be a road network, $\ca S$ a scope mapping on it,
  and $C\subseteq E(G)$ a set of road closures. An $s$-$t$ walk 
  $P$ is {\em enhanced $C$-detour $\ca S$-admissible} if $P$ is simply
  $C^*$-detour $\ca S$-admissible, where $C^*$ is the qc-closure of $C$.
\end{definition}

With that, and using also the definition of a {\em proper} scope mapping
from \cite{HM2011A}, we can now state technical
Proposition~\ref{pro:enhanced-detour-exists}.

In a standard connectivity setting, a graph (road network) $G$ is 
{\em routing-connected\/} if, for every pair of edges $e,f\in E(G)$, there 
exists a walk in $G$ starting with $e$ and ending with $f$.
  A scope mapping $\cal{S}$ of a routing-connected graph $G$ is {\em proper} 
  if, for all $i\in Im({\cal{S}})$, the subgraph $G^{[i]}$ induced by those 
  edges $e\in E(G)$ such that ${\cal S}(e)\geq i$ is routing-connected. 

\begin{proposition}
  \label{pro:enhanced-detour-exists}
  Let $(G,w)$ be a road network, $\ca S$ a proper scope mapping on it,
  and $C\subseteq E(G)$ a set of road closures. Assume $s,t\in V(G)$.
  If there exists a $C$-avoiding $s$-$t$ walk in $(G,w)$, i.e., one not
  containing any edge of $C$, then there also exists an enhanced $C$-detour
  $\ca S$-admissible $s$-$t$ walk there.
\end{proposition}

\subsubsection{Enhanced $C$-Detour $\cal{S}$-Dijkstra's Algorithm.} 

The only difference between this algorithm and the simple one is that we 
have to compute the qc-closure set of $C$. This must be done between step 1 
and step 2 of the simple $C$-detour $\ca S$-Dijkstra's algorithm as follows:

\begin{enumerate}
\item[1.1] \emph{Construction of qc-closure of $C$}
\smallskip

In order to get qc-closure of $C$ we have to run a set of bidirectional 
$\ca S$-Dijkstra algorithms from $s$ to $t$ with a minor modification -- 
we are relaxing all $\cal{S}$-admissible edges (i.e., not only those 
improving distance estimates). In the next run, all edges from $C$ are 
removed in $G$ and the same algorithm is executed again. Edges which 
cannot be reached from $t$ by $\ca S$-admissible walk (in reversed road
network) are quasi-closures. They will be removed from the road network
in the next run. This process continues until all qc-closure set $C^*$
is found. 
\end{enumerate}

By a clever implementation, one can even identify qc-closure set
in a single run of bidirectional $\ca S$-Dijkstra's algorithm. Due to 
lack of space we omit further details in this paper and claim that 
even the implementation of the naive process above requires only a small
number of iterations in practice.

\subsection{Full Detour Admissibility}

In addition to the simple approach of Definition~\ref{def:simple-c-detour}
(and its enhanced version),
we briefly outline a deeper approach which better fits into the overall
idea of scope and comfortable routes, but is technically complicated
and not suited for introductory explanation.
That is why give here an informal outline, while the bare formal definition
is left for the appendix.

\begin{itemize}\def\labelitemi{$\bullet$}
\item
In static $\ca S$-admissibility, we informally
count the $\ca S$-draw value of the travelled subwalk,
and use this information to decide admissibility of edges of restricted
scope.
The same is naturally extended to obstructed vertices and their detours:
Each time a $C$-obstructed vertex $d$ is reached, this lowers the current $\ca
S$-draw value (a better variant of a detour permit)
to one depending on how far $d$ is from the actual closure
(again measured in terms of $\ca S$-draw).

\item
This lowered $\ca S$-draw value then allows exceptional use of low-level
edges for a limited extent since $d$ (in the exactly same way as
low-level edges are allowed near the start $s$).
The same, of course, happens in reverse.

\item
A complication comes from a fact that the whole concept has to be considered
recursively.
This is because lowered $\ca S$-draw value affects $C$-obstructed vertices
further on (there are more of them then, cf.~also secondary detours)
and adds more possibilities of ``detours on detours''.
\end{itemize}

Altogether, the outlined ideas lead to a smooth, though complicated, definition
which nicely incorporates into $\ca S$-Dijkstra's algorithm.
There is also a possibility of a simplified version considering detour
permits only on primary closures---the advantage being in a simple
implementation, virtually almost the same as in
Section~\ref{sec:simple-detours}.
%

\subsection{Experimental Work}

We have implemented a very simple prototype in order to prove good practical
performance of both simple and enhanced $C$-detour $\ca S$-Dijkstra's 
algorithms. Algorithms were implemented in C and compiled using gcc-4.5.1 
without any optimization flags running in a single thread on a machine 
with Intel Core i3 CPU 2.40 GHz with 4 GB RAM. We have used two road 
networks constructed from publicly available TIGER/Line 2010 US roads 
data. Both road network had 10 000 edges and were assigned a scope mapping 
simply according to (corrected) road categories and then it was artificially
balanced to be proper. First road network contains a very small city and its 
rural area, the second contains a~bigger urban area.

We did a set of 500 queries for pairs of randomly distant vertices and we 
placed a few (50) closures randomly on unbounded edges and edge nearby middle 
of the optimal $\ca S$-admissible walk between a pair so that at least one closure
hits the walk. First, we were looking how many quasi-closures will be needed 
in enhanced detour algorithm. The number was very low\footnote{Originally, our 
testing instances contained a lot of degree 2 vertices which was causing a lot of
``fake'' quasi-closures, this problem was solved by merging long chains of such 
edges into a single edge.} -- in average there were only 19 quasi-closures and the 
maximum number of iterations needed to find a qc-closure set was 3. 
Interestingly, there was no significant difference between rural and urban 
testing instance. In average, detour algorithm increased the number of scanned 
vertices just by approx. 8\% and we claim that this number can improved a lot by a 
better implementation. Of course, in case of enhanced algorithm there are more
obstructed vertices given by quasi-closures, but the difference between running
time of both algorithms is neglectable -- 9.2ms and 11.2ms in average. 

\subsection{A Note on Detours in Multi-Staging Scope-Based Approach}

The most important computational aspect of scope lies in the fact that only
the edges of {\em unbounded} scope level $\infty$ matter for global preprocessing
(an idea related to better known {\em reach} \cite{Gutman2004}). Informally, the 
query algorithm of \cite{HM2011A} works in stages: In the \emph{opening cellular
phase}, the road network is locally searched (uni-directional $\cal{S}$-Dijkstra) 
from both start and target vertices until only edges of unbounded scope are 
admissible. Then a small preprocessed ``boundary graph'' is searched by another 
algorithm (e.g.\ hub-based labeling \cite{Abraham2011}) in the \emph{boundary 
phase}. Finally, in the \emph{closing cellular phase}, the scope-unbounded long 
middle section of the route is ``unrolled'' in the whole network.

We remark that the boundary graph will remain static even in the dynamic 
scenario (due to expensive preprocessing), and dynamic changes are mainly 
dealt with in the closing cellular phase. We first remark on the ``only 
negative change'' assumption of our approach (Sec.~\ref{sec:fundamentals}).
This well corresponds with a real-world situation in which just ``bad things
happen on the road'', and the driver thus usually has to find an available
detour, instead of looking for unlikely road improvements. Therefore, we are 
content if our query algorithm finds that an optimal route of the original 
network (wrt.~$w$) is admissible, though not perfectly optimal,\footnote{Note
that a designated detour of a road construction may perhaps turn out faster 
than another previously optimal route.} in the changed network (with $w^*$).
However, when things go worse with $w^*$, then our algorithm works 
efficiently.

\section{Conclusion}

We have outlined the current state of our work on dynamization of the
scope-base route planning technique \cite{HM2011A} for both unexpected and
predictable (to some extend) negative road network changes (closures).  
Our approach is aimed at a proper relaxation of scope admissibility when a 
driver approaches changed road segment, by locally re-allowing nearby roads 
of lower scope level.  At the same time we claim that the computed detour 
minimizes costs and still remains reasonable in terms of scope admissibility.

In a summary, we have shown that a scope-based route planning approach with
cellular preprocessing \cite{HM2011A} can be used not only in static but also 
in dynamic road networks and briefly demonstrated that proof of concept works
well in practice. Our immediate future work in this direction will include 
the following points:
\begin{itemize}
\item an efficient implementation of full $C$-detour algorithm,

\parskip 3pt
\item incorporation of route restrictions and possibly other aspects, and
\item more extensive experimental evaluation of the final implementation.
\end{itemize}

\small 
\bibliographystyle{abbrv}

\begin{thebibliography}{99}

\bibitem{Dijkstra1959}
Edsger~Dijkstra.
\newblock A note on two problems in connection with graphs.
\newblock {\em Numerische Mathematik}, 1:269--271, 1959.

\bibitem{Cooke1966}
Kenneth~L Cooke and Eric Halsey.
\newblock The shortest route through a network with time-dependent internodal
  transit times.
\newblock {\em Journal of Mathematical Analysis and Applications}, 14(3):493 --
  498, 1966.

\bibitem{Hart1972}
Peter~E. Hart, Nils~J. Nilsson, and Bertram Raphael.
\newblock Correction to ``a formal basis for the heuristic determination of
  minimum cost paths''.
\newblock {\em SIGART Bull.}, 1(37):28--29, 1972.

\bibitem{King1999}
Valerie King.
\newblock Fully dynamic algorithms for maintaining all-pairs shortest paths and
  transitive closure in digraphs.
\newblock In {\em Proceedings of the 40th Annual Symposium on Foundations of
  Computer Science}, FOCS '99, pages 81--, Washington, DC, USA, 1999. IEEE
  Computer Society.

\bibitem{Gutman2004}
Ron~Gutman.
\newblock Reach-based routing: A new approach to shortest path algorithms
  optimized for road networks.
\newblock In {\em Proceedings 6th Workshop on Algorithm Engineering and
  Experiments (ALENEX)}, pages 100--111, 2004.

\bibitem{Wagner2004}
Dorothea Wagner, Thomas Willhalm, and Christos~D. Zaroliagis.
\newblock Dynamic shortest paths containers.
\newblock {\em Electr. Notes Theor. Comput. Sci.}, 92:65--84, 2004.

\bibitem{Flinsenberg2004}
Ingrid~C.~M. Flinsenberg.
\newblock {\em Route planning algorithms for car navigation}.
\newblock PhD thesis, Technische Universiteit Eindhoven, 2004.

\bibitem{Wagner2005}
Dorothea Wagner, Thomas Willhalm, and Christos Zaroliagis.
\newblock Geometric containers for efficient shortest-path computation.
\newblock {\em J. Exp. Algorithmics}, 10, December 2005.

\bibitem{Goldberg2005B}
Andrew~V. Goldberg and Chris Harrelson.
\newblock Computing the shortest path: A* search meets graph theory.
\newblock In {\em In Proc. 16th ACM-SIAM Symposium on Discrete Algorithms},
  pages 156--165, 2005.

\bibitem{Sanders2006}
Peter Sanders and Dominik Schultes.
\newblock Engineering highway hierarchies.
\newblock In {\em ESA'06: Proceedings of the 14th conference on Annual European
  Symposium}, pages 804--816, London, UK, 2006. Springer-Verlag.

\bibitem{Delling2007}
Daniel Delling and Dorothea Wagner.
\newblock Landmark-based routing in dynamic graphs.
\newblock In {\em Proceedings of the 6th international conference on
  Experimental algorithms}, WEA'07, pages 52--65, Berlin, Heidelberg, 2007.
  Springer-Verlag.

\bibitem{Schultes2007}
Dominik Schultes and Peter Sanders.
\newblock Dynamic highway-node routing.
\newblock In {\em WEA'07: Proceedings of the 6th international conference on
  Experimental algorithms}, pages 66--79, Berlin, Heidelberg, 2007.
  Springer-Verlag.

\bibitem{Schultes2008}
Dominik Schultes.
\newblock {\em Route Planning in Road Networks}.
\newblock PhD thesis, Karlsruhe University, Karlsruhe, Germany, 2008.

\bibitem{Delling2009}
Daniel Delling and Peter S and Dominik Schultes and Dorothea Wagner. 
\newblock Engineering Route Planning Algorithms. 
\newblock In Algorithmics of Large and Complex Networks. Lecture Notes in Computer Science, pages 117--139. Springer, Berlin, Heidelberg. 2009. 

\bibitem{Abraham2011}
Ittai Abraham, Daniel Delling, Andrew~V. Goldberg, and Renato~F. Werneck.
\newblock A hub-based labeling algorithm for shortest paths in road networks.
\newblock In {\em Proceedings of the 10th international conference on
  Experimental algorithms}, SEA'11, pages 230--241, Berlin, Heidelberg, 2011.
  Springer-Verlag.

\bibitem{HM2011A}
Petr~{Hlin\v{e}n\'y} and Ondrej~{Mori\v{s}}.
\newblock {Scope-Based Route Planning}.
\newblock In {\em ESA'11: Proceedings of the 19th conference on Annual European
  Symposium}, pages 445--456, Berlin Heidelberg, 2011. Springer-Verlag.
\newblock arXiv:1101.3182 (preprint).

\end{thebibliography}

\newpage
\appendix

\section{Classical Dijkstra's Algorithm}
\label{app:dijkstra}

Classical Dijkstra's algorithm solves the single source shortest paths
problem\footnote{Given a graph and a start vertex find the shortest paths 
from it to the other vertices.} in a graph $G$ with a non-negative 
weighting $w$. Let $s\in V(G)$ be the start vertex (and, optionally, let 
$t \in V(G)$ be the target vertex).

\begin{itemize}
\item The algorithm maintains, for all $v \in V(G)$, a 
  {\em (temporary) distance estimate} of the shortest path from $s$ to $v$ 
  found so far in $d[v]$, and a predecessor of $v$ on that path in $\pi[v]$. 
\item The scanned vertices, i.e. those with $d[v] = \delta_w(s,v)$ confirmed, 
  are stored in the set $T$; and the discovered but not yet scanned vertices, 
  i.e. those with $\infty >d[v] \geq \delta_w(s,v)$, are stored in the set $Q$. 
\item The algorithm work as follows: it iteratively picks a vertex $u \in Q$ 
  with minimum value $d[u]$ and relaxes all the edges $(u,v)$ leaving $u$.
  Then $u$ is removed from $Q$ and added to $T$. {\em Relaxing} an edge $(u,v)$
  means to check if a shortest path estimate from $s$ to $v$ may be improved 
  via $u$; if so, then $d[v]$ and $\pi[v]$ are updated. Finally, $v$ is added 
  into $Q$ if is not there already. 
\item The algorithm terminates when $Q$ is empty (or if $t$ is scanned).
\end{itemize}

Time complexity depends on the implementation of $Q$; such as it is 
${\cal{O}}(|E(G)| + |V(G)|\log|V(G)|)$ with the Fibonacci heap.

\medskip 

Dijkstra's algorithm can be used ``bidirectionally'' to solve SPSP 
problem. Informally, one (forward) algorithm is executed from the start 
vertex in the original graph and another (reverse) algorithm is executed from 
the target in the reversed graph. Forward and reverse algorithms can 
alternate in any way and algorithm terminates, for instance, when there is 
a vertex scanned in both directions.

\begin{algorithm}[H]
\caption{~Unidirectional Dijkstra's Algorithm}
\label{alg:dijkstra}
\begin{algorithmic}[1]    
\smallskip
\REQUIRE A road network $(G,w)$ and a start vertex $s \in V(G)$.
\smallskip
\ENSURE For every $v \in V(G)$, an optimal $s-v$ walk in $G$ (or $\infty$).
\end{algorithmic}
\smallskip
\begin{algorithmic}[1]    
 \OLFORALL{$v \in V(G)$}{$d[v] \leftarrow \infty$;~$\pi[v] \leftarrow \bot$;}
 \COMMENT{\hfill// Initialization}
 \STATE $d[s] \leftarrow 0$;~ $Q \leftarrow \{s\}$;~ $T \leftarrow \emptyset$
 \medskip
 \WHILE[\hfill// Main loop]
 {$Q \neq \emptyset \lor t \notin T$}
  \STATE $u \leftarrow  \min_{d[]}(Q)$;~ $Q \leftarrow Q \setminus\{u\}$
  \FORALL[\hfill // Relaxation of $(u,v)$]{$(u,v)\in E(G)$}
   \smallskip
   \OLIF{$d[v] \geq d[u] + w(u,v)$}
   {$d[v] \leftarrow d[u] + w(u,v)$;~ $\pi[v] \leftarrow u$}
   \smallskip
  \ENDFOR
  \STATE $T \leftarrow T \cup \{u\}$
  \COMMENT{\hfill // Vertex u is now scanned}
 \ENDWHILE
 \STATE \textsc{ConstructWalk}$\,(G,d,\pi)$
 \COMMENT{\hfill// Postprocessing -- generating output.}
\end{algorithmic}
\end{algorithm}

\newpage 

\section{$\cal{S}$-Dijkstra's Algorithm}
\label{app:sdijkstra}

Full pseudocode of $\ca S$-Dijkstra's algorithm as presented in \cite{HM2011A}
follows. For better understanding a concept of $\ca S$-reach and its effect is
removed from the Algorithm~\ref{alg:sdijkstra}.

\begin{algorithm}[H]
\caption{~Unidirectional $\cal{S}$-Dijkstra's Algorithm}
\label{alg:sdijkstra}
\begin{algorithmic}[1]    
\smallskip
\REQUIRE A road network $(G,w)$, a scope ${\cal{S}}$ and a start vertex $s \in V(G)$.
\smallskip
\ENSURE For every $v \in V(G)$, an optimal $s$-admissible walk from $s$ to $v$ 
  in $G$ (or $\infty$).
\end{algorithmic}
\smallskip
\underline{\textsc{Relax}$(u,v,\gamma)$} \vskip 2pt
\begin{algorithmic}[1]    
  \IF[\hfill// Temporary distance estimate updated.]{$d[u] + w(u,v) < d[v]$}
    \STATE $Q \is Q \cup \{v\}$
    \STATE $d[v] \is d[u] + w(u,v)$;~$\pi[v] \is u$
  \ENDIF
  \IF[\hfill// Scope admissibility vector updated.]{$d[u] + w(u,v) \le d[v]$}
    \FORALL{$i \in Im({\cal{S}})$}
      \STATE $\sigma_i[v] \leftarrow \min\{\sigma_i[v],\>\sigma_i[u]+\gamma_i\}$ 
    \ENDFOR
  \ENDIF
  \RETURN
\end{algorithmic}

\medskip
\underline{\textsc{${\cal{S}}$-Dijkstra}$(G,w,{\cal{S}},s)$} 
\vskip 3pt
\begin{algorithmic}[1]
 \FORALL[\hfill// Initialization.]{$v \in V(G)$}
  \STATE $d[v] \is \infty$;~$\pi[v] \is \bot$;~
  \COMMENT{\hfill// Distance estimate and predecessor.}
  \STATE $\sigma[v] \leftarrow (\infty,\ldots,\infty)$
  \COMMENT{\hfill// Scope admissibility vector.}
 \ENDFOR
 \STATE $d[s] \is 0$;~ $Q \is \{s\}$;~$\sigma[s] \is (0,\ldots,0)$
 \WHILE[\hfill// Main loop processing all vertices.]{$Q \neq \emptyset$}
  \STATE $u \is  \min_{d[]}(Q)$;~ $Q \is Q \setminus\{u\}$
  \COMMENT{\hfill// Pick a vertex $u$ with the minimum $d[u]$.}
  \FORALL[\hfill// All edges from $u$; subject to]{$f=(u,v)\in E(G)$}
     \IF[\hfill// $s$-admissibility check.]{$\sigma_{{\cal S}(f)}[u]\leq 
       \nu^{\cal{S}}_{{\cal S}(f)}$}
      \FORALL[\hfill// Adjustment to scope admissibility.]{$i\in Im({\cal{S}})$}
        \OLELIF{${\cal{S}}(f)>i$}{$\gamma_i\is w(f)$}{$\gamma_i \is 0$}
      \ENDFOR
      \STATE \textsc{Relax}$(u,v,\gamma)$
      \COMMENT{\hfill// Relaxation of $f=(u,v)$.}
     \ENDIF
  \ENDFOR
 \ENDWHILE
 \STATE \textsc{ConstructWalk}$\,(G,d,\pi)$
 \COMMENT{\hfill// Postprocessing -- generating output.}
\end{algorithmic}
\end{algorithm}

Bidirectional version of $\ca S$-Dijkstra's algorithm is analogous to the 
bidirectional version of classical algorithm -- two searches are execetued, 
one from the start and another from the target in the reverse road network. The
algorithm terminates when there is a vertex scanned in both directions.

\newpage 

\section{Enhanced $C$-Detour $\cal{S}$-admissibility}

In a standard connectivity setting, a graph (road network) $G$ is 
{\em routing-connected\/} if, for every pair of edges $e,f\in E(G)$, there 
exists a walk in $G$ starting with $e$ and ending with $f$. This obviously 
important property can naturally be extended to our scope concept as follows.

\begin{definition}[Proper Scope]
  \label{def:properscope}
  A scope mapping $\cal{S}$ of a routing-connected graph $G$ is {\em proper} 
  if, for all $i\in Im({\cal{S}})$, the subgraph $G^{[i]}$ induced by those 
  edges $e\in E(G)$ such that ${\cal S}(e)\geq i$ is routing-connected, too. 
\end{definition}

\begin{theorem}[\cite{HM2011A}]
  \label{thm:connscope}
  Let $(G,w)$ be a routing-connected road network and let $\cal{S}$ be a~proper
  scope mapping of it. Then, for every two edges $e=(s,x),f=(y,t) \in E(G)$, 
  there exists an $\ca S$-admissible $s$-$t$ walk $P\subseteq G$ such 
  that $P$ starts with the edge $e$ and ends with~$f$
	(i.e., an $e$-$f$ walk).
\end{theorem}

\begin{proof}[of \bfseries Proposition~\ref{pro:enhanced-detour-exists}]
Let $C^*$ denote the qc-closure of $C$ (cf.~Def.~\ref{def:enhanced-detour}).
By definition, an $s$-$t$ walk is $C$-avoiding iff it is $C^*$-avoiding.
By Theorem~\ref{thm:connscope}, there is an $\ca S$-admissible $s$-$t$ walk
$P\subseteq G$, and we take an optimal such~$P$.

Let $f=(u,v)\in E(P)$ be the first edge of $P$ such that $f\in C^*$
(if such one does not exist, then we are done with $P$).
Observe that there is another edge $f'=(u,v')\not\in C^*$
from which $t$ can be reached on a $C^*$-avoiding walk:
If $u=s$, then this follows from the assumption of existence of a
$C$-avoiding $s$-$t$ walk.
Otherwise, let $f_1=(u_1,u)$ be the edge preceding $f$ on $P$.
Since $f_1\not\in C^*$ by our choice of $f$, there must be a
$C^*$-avoiding $u_1$-$t$ walk starting with $f_1$, 
and $f'$ can be chosen as the second edge on it.
Symmetric claim holds, of course, in reverse from $t$.

Now, the vertex $u$ is $C^*$-obstructed for initial $draw^{\ca S}(P^{su})$
and target $t$, with obstruction level~$0$ (cf.~Def.~\ref{def:c-obstructed}).
Hence $f'$ will be in accordance with point (iii.) of
Def.~\ref{def:simple-c-detour}.
We take a new $s$-$t$ walk $P'$ with prefix $P^{su}.f'$\,---this walk
continues from $f'$ to $t$ with an optimal $\ca S$-admissible $u$-$t$ walk
that additionally never decreases the scope level at the start
(such a walk exists thanks to proper scope mapping, analogically to the
proof of Theorem~\ref{thm:connscope}).
If, again, $P'$ intersects $C^*$, then we repeat the above argument.
In a finite number of steps, we reach the conclusion.
Obviously, the constructed walk may be far away from optimality, but that is
not the objective of this claim.
\qed\end{proof}

\section{Full Detour Admissibility: the Definition}

\begin{definition}[Full $C$-detour $\ca S$-admissibility]
  \label{def:full-c-detour}
  Let $(G,w)$ be a road network, $\ca S$ a scope mapping on it,
  and $C\subseteq E(G)$ a set of road closures. An $s$-$t$ walk 
  $P=(s=u_0,e_1,\dots e_k,u_k=t)\subseteq G$ is {\em fully $C$-detour 
  $\ca S$-admissible} if $E(P)\cap C=\emptyset$
  and the following are true:
  \begin{enumerate}[i.]
    \parskip 2pt
  \item
  There exist indices $a_0=0<a_1<\dots<a_p<k$, and $c_0=k>c_1>\dots>c_q>0$ in
  reverse; here to avoid nested indexing, we shortly write $d_i=u_{a_i}$
  and $d'_i=u_{c_i}$.
  Moreover, there is a set\footnote{The meaning of $B$ is technical:
	This set presents the ``breakpoints'' on $P$ at which we switch from
	considering $\ca S$-admissibility straight to considering it in
	reverse (i.e., from valueing $\ca S$-draw from the last obstructed vertex
	to borrowing it till the next obstructed vertex in reverse).
	Switch back happens automatically after finishing the detour.}
  $B\subseteq V(P)$ such that, between every two succeeding $d_i$ and $d_j'$
  on $P$, there is one selected vertex between them in~$B$.
  \item
  $\vec\pi[d_0]=\vec0$ is the zero $\ca S$-vector.
  \item
  For each $i>0$, this $d_i$ is $C$-obstructed for the target $t$,
  and $\vec\pi[d_i]$ shortly denotes the corresponding $C$-obstruction state
  of~$d_i$.
  Specially, if $P^{\,d_{i-1}\!d_i}$ (the short subwalk from $d_{i-1}$ to $d_i$)
  avoids $B$, then $d_i$ is $C$-obstructed for the initial $\ca S$-vector
  $\vec\omega[d_i]$ and $t$ (otherwise, $\vec\omega[d_i]=\vec\infty$):~
  $\vec\omega[d_i]=\vec\pi[d_{i-1}]+ draw^{\ca S}(Q)$ where $Q$ is an
  optimal $(d_{i-1},\vec\pi[d_{i-1}])$-admissible $d_{i-1}$-\,$d_i$ walk.
  \item
  (ii.) and (iii.) are analogously formulated for $d_i'$ in reverse.
  \item
  For each $e_{m+1}=(u_{m},u_{m+1})\in E(P)$, (at~least) one of the following holds:
    \begin{itemize}
    \item $\ca S(e_{m+1})=\infty$.
    \item $\ca S(e_{m+1})=\ell<\infty$, and there exists $0\leq j\leq m$ such that
    $j=a_i$, the subwalk $P^{d_iu_{m}}$ from $d_i$ to (including) $u_m$ is disjoint from
    $B\cup\{d_1',\dots,d_q'\}$, and the following is fulfilled:
    It is \mbox{$\vec\pi_\ell[d_i]+\big[draw^{\ca S}(Q)\big]_\ell\leq
	 \nu^{\cal{S}}_{\ell}$} for some optimal
    $(d_{i},\vec\pi[d_{i}])$-admissible $d_{i}$-$u_m$ walk~$Q$.
    \item The previous holds in reverse for some $m+1\leq j\leq k$ such that
    $j=c_i$.
    \end{itemize}
  \end{enumerate}
\end{definition}

\end{document}